\documentclass[10pt]{article}
\usepackage[T1]{fontenc}
\usepackage[english]{babel}
\usepackage{amsmath,amssymb}
\usepackage{amsthm}
\usepackage{enumerate}
\usepackage{eucal,mathrsfs}

\theoremstyle{plain}
\newtheorem{theorem}{Theorem}
\newtheorem{lemma}[theorem]{Lemma}

\theoremstyle{definition}

\theoremstyle{remark}

\newcommand{\abrtitle}{Space-time covariance functions with compact support}

\newcommand{\eee}{{\rm e}}

\newcommand{\LL}{\mathcal{L}}

\newcommand{\comp}{\mathbb{C}}

\newcommand{\real}{\mathbb{R}}

\renewcommand{\Re}{\ensuremath{{\rm Re\,}}}

\newcommand{\C}{\mathbb{C}}
\newcommand{\N}{\mathbb{N}}
\newcommand{\R}{\mathbb{R}}

\newcommand{\s}{\mathbb{S}}

\usepackage{color}

\begin{document}
\allowdisplaybreaks \setlength{\baselineskip}{20pt}
\begin{titlepage}
\begin{center}
\textbf{\large \abrtitle}\\[6ex]

\vspace{2cm} {\itshape

 Viktor P. Zastavnyi \\
    \textsf{Donetsk National University} \\
    \textsf{Department of Mathematics} \\
    \textsf{Universitetskaya str. 24, Donetsk, 340001, Ukraine} \\
    \texttt{zastavn@rambler.ru} \\

\bigskip \bigskip

 Emilio Porcu \\ 
    \textsf{University Jaume I of Castell\'{o}n} \\
    \textsf{Department of Mathematics} \\
    \textsf{Campus Riu Sec} \\
    \textsf{E-12071 Castellón, Spain} \\
    \texttt{porcu@mat.uji.es}

 \bigskip \bigskip
 \bigskip \bigskip

}
\end{center}
 \bigskip \bigskip
\begin{center}
\tt {\bf \tt Abstract}
\end{center}
{\small \tt We characterize completely the Gneiting class
\cite{Gneiting_2002} of space-time covariance functions and give
more relaxed conditions on the involved functions. We then show
necessary conditions for the construction of compactly supported
functions of the Gneiting type. These conditions are very general
since they do not depend on the Euclidean norm. Finally, we discuss
a general class of positive definite functions, used for
multivariate Gaussian random fields. For this class, we show
necessary criteria for its generator to be compactly supported. \\

}

\title{\bfseries\abrtitle}
 \bigskip \bigskip \vspace*{.5cm}
\author{Viktor P. Zastavnyi and Emilio Porcu}
\end{titlepage}
\date{}
\maketitle

\section{Introduction}

Recent literature persistently emphasizes the use of approximation
methods and new methodologies for dealing with massive spatial data
set. When dealing with spatial data, calculation of the inverse of
the covariance matrix becomes a crucial problem. For instance, the
inverse is needed for best linear unbiased prediction (alias
kriging), and is repeatedly calculated in the maximum likelihood
estimation or the Bayesian inferences. Thus, large spatial sample
sizes traduce into big challenges from the computational point of
view.

A natural idea that made proselytes in the last year is to make the
covariances exactly zero after certain distance so that the
resulting matrix has a high proportion of zero entries and is
therefore a sparse matrix. Operations on sparse matrices take up
less computer memories and run faster. However, this should be done
in a way to preserve positive definiteness of the resulting
covariance matrix. The idea goes under the name of covariance {\em
tapering}, by meaning that the true covariance is multiplied
pointwise with a compactly supported and radial correlation
function. This operation is technically justified by the fact that
the Schur product preserves positive definiteness.

The effects of tapering in terms of estimation and interpolation
have been recently inspected by \cite{Du:Zhang:Mandrekar}, where
general conditions are given in order to ensure that tapering does
not affect the efficiency of the maximum likelihood estimator. For
spatial interpolation, \cite{Furrer:Genton:Nychika:2006} show that
under some regularity conditions, tapering procedures yield
asymptotically optimal prediction. In order to assess these
properties, the asymptotic framework adopted by the authors is of
the infill type, and the tool allowing to evaluate the performances
of tapering is the equivalence of Gaussian measures, for which a
comprehensive theory can be found in the seminal work by Yadrenko
\cite{Yadrenko}.

These points fix very briefly the state of the art and we refer the
reader to \cite{Du:Zhang} for an excellent survey on the topic.

Although tapering has been well understood in the spatial framework,
there is nothing done, to the knowledge of the authors, for the
spatio-temporal case. In particular, the use of tapering is at least
questionable in space-time, since the same type of asymptotics does
not apply and thus it is not easy to evaluate its performances.

But a deeper look at this problem also highlights the non existence,
in the literature, of space-time covariance functions that are
compactly supported over space, time or both. These facts motivate
the research documented in this manuscript.

We deal with challenges related to space-time covariance functions.
If spatial data set can be massive, one can imagine how the
dimensionality problem affects space-time estimation and
interpolation. This problem may be faced on the base of two
perspectives that can be illustrated through the celebrated T.
Gneiting class of covariance functions \cite{Gneiting_2002}: for
$(x,t) \in \R^{d+l}$, the function
\begin{equation}
\label{gneiting_class} (x,y) \mapsto K(x,t):= h(\|t\|^2)^{-d/2}
\varphi \left ( \frac{\|x\|^2}{h(\|t\|^2)}\right )
\end{equation}
is positive definite, for $\varphi$ completely monotone on the
positive real line and $h$ a Bernstein function. For $l=1$, the
function above is a stationary and nonseparable space-time
covariance. This function has been persistently used by the
literature and a Google scholar search highlights that currently
there are over 90 papers where this covariance has been used for
applications to space-time data.

If there are many observations over space, time or both, then the
use of this function would be questionable for the computational
reasons exposed above. A more intriguing perspective is to consider
 a function of the Gneiting type, but replacing the generator
$\varphi$ in equation (\ref{gneiting_class}) with a compactly
supported function, and inspecting the conditions ensuring that
permissibility is preserved on some $d$-dimensional Euclidean space.
The results are illustrated in the following sections.

An auxiliary result of independent interest is also given: we
characterize completely the Gneiting class and give more general
conditions for its permissibility.

The ratio mentis of this paper leads then to consider a general
class of covariances, originally proposed in Porcu {\em et al.}
\cite{Porcu:Gregori:Mateu2006} and more recently in \cite{Apa-Gen}.
Both groups of authors show that is class of covariances can be very
versatile since it can be used for two-fold purposes: on the one
hand, it can be effectively used to deal with zonally anisotropic
structures, on the other hand it can be adapted to represent the
covariance mapping associated to a multivariate random field, which
is highly in demand since there are very few models with these
characteristics \cite{gnei:schlath}.

As a conclusion to the preludium, the plan of the paper is the
following: in Section 2 we present basic facts about positive and
negative definite functions. Section 3 characterizes completely the
Gneiting class, for which only sufficient conditions were known
until now. In Section 4 we present necessary conditions for
compactly supported covariances of the Gneiting type. Similar
results are obtained in Section 5 for the multivariate class of
cross-covariances proposed in \cite{Porcu:Gregori:Mateu2006}.

\section{Preliminaries}

This section is largely expository and contains basic facts and
information needed for a self-contained exposition. We shall
enunciate the concepts of positive and negative definiteness, as
well as the material related to them, working with linear spaces and
subspaces. The space-time notation will be used only when necessary
for a clearer exposition of results.

For $E$ a real linear space, we denote by {\rm FD}$(E)$ the set of
all linear finite-dimensional subspaces of $E$. If $\dim E=n\in\N$
and $e_1,\ldots,e_n$ are basis in $E$, then
\begin{eqnarray*}
f\in C(E) & \iff & f(x_1e_1+\ldots+x_ne_n)\in C(R^n) \\
&{\rm and}& \\
f\in L(E) &\iff & f(x_1e_1+\ldots+x_ne_n)\in L(R^n).
\end{eqnarray*}

Also, we call $C_0(E)$ the set of all function $f\in C(E)$ such that
$f$ has compact support. If $\dim E=\infty$, then
  $f\in C(E)$ $\iff$ $f\in C(E_0)$ $\forall E_0\in {\rm FD}(E)$.

A complex-valued function $f:\ E\to\C$ is said to be positive
definite on $E$ (denoted hereafter $f\in \Phi(E)$) if for any finite
collection of points $\{\xi_i\}_{i=1}^n \in E$ the matrix
$(f(\xi_i-\xi_j))_{i,j=1}^n$ is positive definite, \emph{i.e.}\
$$
    \text{for all\ \ } a_1, a_2, \ldots, a_n\in\comp\::\qquad \sum_{i,j=1}^n a_i f(\xi_i-\xi_j) \overline{a}_j \geq 0.
$$

It is well known that the family of positive definite functions is a
convex cone which is closed under addition, products, pointwise
convergence and scale mixtures. Briefly, we have the following
properties.

Let $f$, $f_{i}\in\Phi(E)$, $i\in\N$. Then:
\begin{description}
\item[] 1. $|f(x)|\le f(0)$,
 $\overline{f(-x)}=f(x)$, $|f(x)-f(h)|^2\le 2f(0){\rm Re}(f(0)-f(x-h))$, $x$, $h\in E$;
 \item[] 2.  $\lambda_{1}f_{1}+\lambda_{2}f_{2}$ with $\lambda_i\ge 0$,
$\bar{f}$, $\Re f$, $f_1f_2\in\Phi(E)$;
 \item[] 3.  if, for all $x\in E$, the finite limit $\lim\limits_{n\to\infty}
f_{n}(x)=:g(x)$ exists, then $g\in\Phi(E)$;
\item[] 4.  for any linear
operator $A:\ E_1\to E$ the function $f\circ A$ belongs to
$\Phi(E_1)$; in  particular, $f\in\Phi(E_1)$ for any linear subspace
$E_1$ from $E$.
\end{description}
Let $E=\R^n$. The celebrated Bochner's theorem establishes a one to
one correspondence between continuous positive definite functions
and the Fourier transform of a positive and bounded measure, {\em
i.e.} $f(x)=F_n (\mu(u))(x)$. If $\mu$ is absolutely continuous with
respect to the Lebesgue measure, than $d \mu(u)= \widehat{f}(u)du$,
for  $\widehat{f}$ nonnegative. This can be rephrased in the
following way: if $f\in C(\R^n)\cap L(\R^n)$, then $f\in\Phi(\R^n)$
if and only if
$$
\widehat{f}(u)=F^{-1}_n(f)(u):=\int_{\R^n}e^{i(u,x)}f(x)\ dx\ge 0,
\quad u\in\R^{n},
$$
for $(\cdot,\cdot)$ the usual dot product. The function
$\widehat{f}$ is called spectral density or Fourier pair associated
to $f$.

If $f$ is a radially symmetric and continuous function depending on
the squared Euclidean norm $\|\cdot\|_2^2$, i.e.
$f(x)=\varphi(\|x\|_2^2)$, $\varphi\in C_{[0,+\infty)}$, then the
Fourier transform above simplifies to the Bessel integral (if in
addition $f\in L(\R^n)$)
\begin{equation} \label{eq:hankel_euclidean}
g_n(s):=\int_{0}^{+\infty}\varphi(u^2)u^{n-1}j_{\frac
n2-1}(su)\,du\,,
\end{equation}
 where $j_{\lambda}(u):=\frac{J_{\lambda}(u)}{u^{\lambda}}$, with
  $J_{\lambda}$ a Bessel function of the first kind. Thus $f \in
  \Phi(\R^n)$, for some $n \in \mathbb N$ and for $f$ radially symmetric, if and only if $g_n(u) \geq
  0$  $\forall u >0$.

A function $f:]0,\infty[\to\real$ is called \emph{completely
monotone}, if it is arbitrarily often differentiable and
\begin{equation*}
    (-1)^nf^{(n)}(x)\ge 0\text{\ \ for\ \ } x>0,\; n=0,1,\ldots.
\end{equation*}
By Bernstein's theorem the set $M_{(0,\infty)}$ of completely
monotone functions coincides with that of Laplace transforms of
positive measures $\mu$ on $[0,\infty[$, \emph{i.e.}\
\begin{equation*}\label{eq:bernstein}
    f(x)=\LL \mu(x)=\int_{[0,\infty[} \eee^{-xt}\,d\mu(t),
\end{equation*}
where we only require that $\eee^{-xt}$ is $\mu$-integrable for any
$x>0$. $M_{(0,\infty)}$ is a convex cone which is closed under
addition, multiplication and pointwise convergence.

The connection with the function $g_n(\cdot)$ gives the celebrated
Schoenberg (1939) theorem by which a radial function
 $f(x)=\varphi(\|x\|_2^2)$, $\varphi\in
C_{[0,+\infty)}$, belongs to $\Phi(\R^n)$ for all $n \in \mathbb N$
if and only if $\varphi$ is completely monotone on the positive real
line, and in this case the Bessel integral in equation
(\ref{eq:hankel_euclidean}) reduces to a Gaussian mixture. Finally,
a Bernstein function is a positive function that is infinitely often
differentiable and whose first derivative is completely monotone.
For a more detailed exposition on these facts the reader is referred
to \cite{porcu-schilling}.

In this paper we shall be also dealing with functions depending not
on the Euclidean norm but on some homogeneous continuous function
$\rho: E \to \R$ such that $\rho(tx)=|t|\rho(x)$ $\forall t\in\R,
x\in E$ and $\rho(x)>0$, $x\ne 0$. If  $\varphi\in C_{[0,+\infty)}$
and $\int_{0}^{+\infty}|\varphi(t^2)|t^{n-1}\;dt<+\infty$,
 then we have that
$\varphi \circ \rho^2\in \Phi(\R^n)$ if and only if the function
\begin{equation}\label{Gm}
 \R^n \ni v \mapsto G_n(v):=\int_{\R^n}\varphi(\rho^2(y))e^{i(y,v)}\,dy
 \end{equation}
is nonnegative for all $v \in \R^n$. If $\rho$ is the Euclidean
norm, then the functions $G_n(\cdot)$ and $g_n(\cdot)$ are related
by the well know equality $G_n(v)=(2\pi)^{\frac{n}{2}}g_n(||v||_2)$.

Finally, a complex-valued function $h:\ E\to\C$ is called
(conditionally) negative definite on $E$ (denoted $h\in N(E)$
hereafter) if the inequality
$$
\sum^{n}_{k,j=1}c_{k}\bar{c}_{j}h(x_{k}-x_{j})\le 0
$$
is satisfied for any finite systems of complex numbers
$c_1,c_2,...,c_n$, $\sum_{k=1}^{n}c_k=0$, and points $x_1,...,x_n$
in $E$.

Let $\{Z(\xi), \xi \in \R^n \}$ be a continuous weakly stationary
and Gaussian random field (RF for short). The associated covariance
function $f:\R^n \to \real$ is positive definite. This can be
rephrased by saying that positive definiteness of a candidate
continuous function $f:\R^n \to\real$ is sufficient condition for
the existence of a continuous weakly stationary and Gaussian RF
having $f(\cdot)$ as covariance function.

 If, additionally, $f(\cdot)$ is
\emph{radially symmetric}, the associated Gaussian RF is called
\emph{isotropic}. Isotropy and stationarity are independent
assumptions but throughout the paper we shall assume both in order
to keep things simple.

To complete the picture, the variance of the increments of an
intrinsically stationary Gaussian RF is called variogram. For two
points of $\R^n$, say $\xi_i$, $i=1,2$, we have that $\mathbb{V}{\rm
ar} \left (Z(\xi_2)-Z(\xi_1) \right ):= \gamma(\xi_2-\xi_1)$. The
mapping $\gamma(\cdot): \R^n \to \R$ is conditionally negative
definite. The additional property of isotropy is then analogously
defined as before.

\section{Complete Characterization of the Gneiting class}


\begin{lemma} \begin{description} \item[]
\item[] {\rm i.} $f\in\Phi(E)\iff f\in\Phi(E_0)$ $\forall E_0\in {\rm FD}(E)$. \item[] {\rm ii.}
If $\dim E=n\in\N$ then $f\in\Phi(E)\iff fg\in\Phi(E)$  $\forall
g\in\Phi(E)\cap C_0(E)$. \end{description}
\end{lemma}
\begin{proof}
\begin{description} \item[]
\item[] {\rm i.} The necessity is obvious. As for the sufficiency, for $n\in\N$ and
$x_1,...,x_n$ in $E$, we have that $x_1,...,x_n\in E_0$ - the linear
span of these elements. Obviously $\dim E_0\le n$.
\item[] {\rm ii.} Again, the necessity is obvious. For the sufficiency, let $e_1,\ldots,e_n$ be basis in $E$. Then we
take
$g(x_1e_1+\ldots+x_ne_n)=(1-\varepsilon|x_1|)_+\cdot\ldots\cdot(1-\varepsilon|x_n|)_+$
and $\varepsilon \downarrow0$. The proof is completed.
\end{description}
\end{proof}

  \begin{lemma}\label{le1}
Let the next conditions be satisfied:
\begin{enumerate}
\item $h,b\in C(E)$ and $h(t)>0$  $\forall t\in E$.
\item $\varphi\in C_{[0,+\infty)}$ and for the some $m\in\N:$
$\int_{0}^{+\infty}|\varphi(u^2)|u^{m-1}\,du<+\infty$.
\item $\rho\in C(\R^m)$, $\rho(tx)=|t|\rho(x)$ $\forall t\in\R, x\in\R^m$ and $\rho(x)>0$, $x\ne 0$.
\end{enumerate}
 Then $$
 K(x,t):=b(t)\varphi\left(\frac{\rho^2(x)}{h(t)}\right)\in \Phi(\R^m\times E)\iff
b(t)(h(t))^{\frac m2}G_m(\sqrt{h(t)}v)\in \Phi(E)\;\forall
v\in\R^m\,,
 $$
 with $G_m(\cdot)$ defined in equation (\ref{Gm}).
\end{lemma}

\begin{proof}
Observe that $\varphi\left(\rho^2(x)\right)\in L(\R^m)$. We have
that
\begin{eqnarray*}
K(x,t)\in \Phi(\R^m\times E)&\iff& K(x,t)\in \Phi(\R^m\times
E_0)\;\; \forall E_0\in {\rm FD}(E) \\ &\iff& K(x,t)g(t)\in
\Phi(\R^m\times E_0)\;\; \forall E_0\in {\rm FD}(E),\forall g\in
\Phi(E_0)\cap C_0(E_0) \\ &\iff & \iint_{\R^m\times
E_0}K(x,t)g(t)e^{i(x,v)}e^{i(t,u)} dx dt
   \geq 0 \\ && \forall E_0\in {\rm
FD}(E),\forall g\in \Phi(E_0)\cap C_0(E_0)\;,\;\forall v\in\R^m,u\in
E_0.
\end{eqnarray*}
As for the last integral, a change of variables of the type
$x=\sqrt{h(t)}y$ yields that the last inequality is equivalent to
$$
 \int_{E_0}g(t)b(t)(h(t))^{\frac m2}G_m(\sqrt{h(t)}v)e^{i(t,u)}dt\ge 0\;,\;\forall v\in\R^m,u\in
 E_0,
$$
which holds if, and only if $\forall g\in \Phi(E_0)\cap
C_0(E_0),\forall
 v\in\R^m$, we have
\begin{eqnarray*}
 && g(t)b(t)(h(t))^{\frac m2}G_m(\sqrt{h(t)}v)\in \Phi(E_0)\;\;
 \forall E_0\in {\rm FD}(E), \quad  \\
 &\iff& b(t)(h(t))^{\frac m2}G_m(\sqrt{h(t)}v)\in \Phi(E)\;\forall
v\in\R^m\,.
\end{eqnarray*} The proof is completed.
\end{proof}

The following result gives a complete characterization of the
Gneiting class, with the additional feature that only negative
definiteness of the function $h$ is required, whilst Gneiting's
assumptions are much more restrictive as it is required that
$h^{\prime}$ is completely monotone on the positive real line.
Furthermore, we give a simple proof of this result and we defer it to next section for the reasons that will become apparent throughout the paper.

\begin{theorem}\label{th3}
Let  $h\in C(E)$, $h(t)>0$  $\forall t\in E$. Let $d\in\N$. The
following statements are equivalent:
\begin{enumerate}
\item $
 K(x,t):=(h(t))^{-\frac d2}\varphi\left(\frac{||x||_2^2}{h(t)}\right)\in \Phi(\R^d\times E)
 \;\;\forall \varphi\in C_{[0,+\infty)}\bigcap M_{(0,+\infty)}\,.
$
\item $e^{-\lambda h(t)}\in \Phi(E)$ $\forall \lambda>0$.
\end{enumerate}
 \end{theorem}

Let us consider examples of functions $h$ for which the statement 2
in Theorem \ref{th3} holds.
\begin{description}
\item[Example 1] {\rm
 Let $h(t)=||t||_p^\alpha+c$, $c>0$, $0<p\le+\infty$, $\alpha\ge
0$, $t=(t_1,\ldots,t_n)\in\R^n$, where
$||t||_p^p=\sum_{k=1}^{n}|t_k|^p$, $0<p<\infty$, and
$||t||_{\infty}=\sup_{1\le k\le n}|t_k|$. Then
   $$
   e^{-\lambda h(t)}\in \Phi(\R^n)\;\forall \lambda>0
   \iff e^{-||t||_p^\alpha}\in\Phi(\R^n)  \iff 0\le \alpha\le \alpha(l_{p}^{n})\;,$$ where
\begin{equation}
 \alpha(l_{p}^{n})= \left\{ \begin{array}{ll}
2\ &{\rm if }\ n=1,\ 0<p\le\infty;\\
p\ &{\rm if }\ n\ge 2,\ 0<p\le 2;\\
1\ &{\rm if }\ n=2,\ 2<p\le\infty;\\
0\ &{\rm if }\ n\ge 3,\ 2<p\le\infty. \end{array} \right.
\end{equation}
For $0<p\le 2$, we get Schoenberg's result. The other two cases have
been investigated by Koldobsky~\cite{Kold91} in 1991 and
Zastavnyi~\cite{Zast91,Zast92,Zast93}  in 1991 ($2<p\le\infty$,
$n\ge 2$). Finally, Misiewiez \cite{Mis89} gave the last result in
1989 ($p=\infty$, $n\ge 3$). }
\item[]
\item[Example 2] {\rm
 If $\rho(t)$ is a norm on $\R^2$, then
$e^{-\rho^\alpha(t)}\in\Phi(\R^2)$ for all $0\le\alpha\le 1$. This
is a well-know fact (see, for example, \cite{JMVA}).
 Therefore
$e^{-\lambda h(t)}\in \Phi(\R^2)$ $\forall \lambda>0$, where
$h(t)=\rho^\alpha(t)+c$, $0\le\alpha\le 1$, $c>0$.
 } \item[] \item[Example 3]  {\rm
  Let  $\psi(s)\in\R$ $\forall s>0$. Then, it is well known that
   $$e^{-\lambda\psi}\in M_{(0,+\infty)}\;\forall \lambda>0\iff\psi'\in M_{(0,+\infty)}.$$
Gneiting \cite{Gneiting_2002} proves the following: if $\psi\in
C_{[0,+\infty)}$,  $\psi(s)>0$ $\forall s\ge 0$,
 and $\psi'\in M_{(0,+\infty)}$,
then $e^{-\lambda h(t)}\in\Phi(\R^n)$ for all $\lambda>0$, $
n\in\N$, where   $h(t):=\psi(||t||_2^2)$ and, hence,
  $$K(x,t):=(\psi(||t||_2^2))^{-\frac d2}\varphi\left(\frac{||x||_2^2}{\psi(||t||_2^2)}\right)\in \Phi(\R^d\times \R^n)
   \;\forall \varphi\in C_{[0,+\infty)}\bigcap M_{(0,+\infty)}\,,\,d\in\N\,.$$

  } \item[] \item[Example 4] By celebrated Schoenberg's Theorem \cite{Scho39}, if $h(-t)=h(t)$ $\forall t\in E$, then
 $$
 h(t)\in N(E)\iff e^{-\lambda h(t)} \in \Phi(E)\; \forall \lambda>0\,.
$$ \item[] \item[Example 5] {\rm
  Let  $g\in \Phi(E)$, $g(-t)=g(t)$ for all $t\in E$ and $h(t):=g(0)-g(t)+c$, $c>0$.
  Then $h(t)>0$ $\forall t\in E$, $h\in N(E)$ and, hence,
  $e^{-\lambda h(t)}\in\Phi(E)$ for all $\lambda>0$.
  }

\end{description}

\section{Necessary conditions for compactly supported functions of the Gneiting type}

From now on let us write $\s^{d-1}:=\{x\in\R^d:||x||_2=1\}$ for the
sphere of $\R^d$.

\begin{theorem}\label{th1}
Let the next conditions be satisfied:
\begin{description}
\item[] 1) $h\in C(E)$, $h(t)>0$  $\forall t\in E$ and $h(t)\not\equiv h(0)$ on $E$.
\item[] 2) $\varphi\in C_{[0,+\infty)}$, $\varphi(0)>0$.
\item[] 3) For $d \in \mathbb N$, $\rho\in C(\R^d)$,
$\rho(tx)=|t|\rho(x)$ $\forall t\in\R, x\in\R^d$ and $\rho(x)>0$,
$x\ne 0$.
\item[] 4) $
 K(x,t):=(h(t))^{-\frac d2}\varphi\left(\frac{\rho^2(x)}{h(t)}\right)\in \Phi(\R^d\times E)
 $.
\end{description}
Then:
\begin{description}
\item[1.] $(h(t))^{-\frac d2}\in \Phi(E)$ and $\varphi\left(\rho^2(x)\right)\in\Phi(\R^d)$.
\item[2.] If there exists a $n\in\N\bigcap[1,d]$ such that
  $\int_{0}^{+\infty}|\varphi(u^2)|u^{n-1}\,du<+\infty$,
  then  $\forall m=1,\ldots,n$ and $v\in\R^m$ the function
  $s \mapsto f_{m,v}(s):=s^{m-d}G_m(sv)$, with $G_m(\cdot)$ as
  defined in (\ref{Gm}),
 is decreasing on $(0,+\infty)$. Furthermore, $f_{m,v}(+\infty)=0$ for $v\ne 0$.
\item[3.] If $\int_{0}^{+\infty}|\varphi(u^2)|u^{d-1}\,du<+\infty$,
then
 $G_d(0)>0$.
 If, in addition, $G_d$ is real-analytic, then $\forall v\in\R^d$, $v\ne 0$ the function $s \mapsto f_{d,v}(s):=G_d(sv)$
is strictly decreasing on $[0,+\infty)$ and $G_d(v)>0$ $\forall
v\in\R^d$.
\item[4.] If
  $\int_{0}^{+\infty}|\varphi(u^2)|u^{d+1}\,du<+\infty$, then
  $\alpha_1(v):=\int_{\R^d}\varphi(\rho^2(y))(y,v)^2\,dy\ge 0$ $\forall v\in\s^{d-1}$ and
  $\beta_1:=\int_{\R^d}\varphi(\rho^2(y))||y||_2^2\,dy\ge 0$.
  Furthermore, $\alpha_1(v)\equiv 0$ on $\s^{d-1} \iff \beta_1=0$.
   If, in addition, $\beta_1> 0$,
   then
   $e^{-\lambda h(t)} \in \Phi(E)\; \forall \lambda>0$.
\item[5.] If $\int_{0}^{+\infty}|\varphi(u^2)|e^{\varepsilon
u}\,du<+\infty$ for some $\varepsilon >0$ (for example, when
$\varphi$ has compact support), then $\exists p\in\N:$ $e^{-\lambda
h^p(t)} \in \Phi(E)\; \forall \lambda>0$. In practice, for $p$ it is
possible to take one of the following numbers:
$$
  p(v):=\min\left\{k\in\N:\;\alpha_k(v)=\int_{\R^d}\varphi(\rho^2(y))(y,v)^{2k}\,dy\ne 0\right\}
  \;,\;v\in\s^{d-1}\;,
$$
$$
  q:=\min\left\{k\in\N:\;\beta_k=\int_{\R^d}\varphi(\rho^2(y))||y||_2^{2k}\,dy\ne 0\right\}
  \;.
$$
The function $p(\cdot)$ is bounded on $\s^{d-1}$ and
$q=\min\limits_{v\in\s^{d-1}}p(v)$.
\end{description}
\end{theorem}
\begin{proof}
 The statement {\bf 1} is obvious.

Let us prove the statement {\bf 2.} By Lemma \ref{le1}, we have
 $$
 F_{m,v}(t):=(h(t))^{\frac{m-d}{2}}G_m(\sqrt{h(t)}v)\in \Phi(E)\;,\;\forall
 m=\overline{1,n}\;,\;v\in\R^m\;.
 $$
 Hence, $F_{m,v}(0)=(h(0))^{\frac{m-d}{2}}G_m(\sqrt{h(0)}v)\ge 0$ and $|F_{m,v}(t)|\le F_{m,v}(0)$, $t\in E$.
 Therefore $G_m(v)\ge 0$, $v\in\R^m$, and
 $$
 (sh(t))^{\frac{ m-d}{2}}G_m(\sqrt{h(t)}sv)\le
 (sh(0))^{\frac{ m-d}{2}}G_m(\sqrt{h(0)}sv)
 \;,\;\forall
 m=\overline{1,n}\;,\;v\in\R^m\;,\;s> 0\;,\;t\in E\,.
 $$
 The latter inequality is equivalent to
$$
f_{m,v}\left(\sqrt{\frac{h(t)}{h(0)}}\cdot s\right)\le
f_{m,v}\left(s\right) \;,\;\forall
 m=\overline{1,n}\;,\;v\in\R^m\;,\;s> 0\;,\;t\in E\,.
$$
 Since $(h(t))^{-\frac d2}\in \Phi(E)$, then $h(t)\ge h(0)$, $t\in E$.
 Since $h(t)\not\equiv h(0)$ on $E$, then there exists a point $t_0\in E$ such
 that $q:=\sqrt{\frac{h(t_0)}{h(0)}}>1$.
 By the intermediate values Theorem $\forall \alpha\in [1,q]$ $\exists \xi\in E:$ $\sqrt{\frac{h(\xi)}{h(0)}}=\alpha$.
 Therefore,
 $f_{m,v}(\alpha s)\le f_{m,v}(s)$ for all $s>0$ and $\alpha \in[1,q]$.
 Hence, $f_{m,v}(\alpha^2 s)\le f_{m,v}(\alpha s)\le f_{m,v}(s)$ for all $s>0$ and $\alpha \in[1,q]$.
 Thus, $f_{m,v}(\alpha^p s)\le f_{m,v}(s)$ for all $s>0$,  $\alpha \in[1,q]$ and $p\in \N$.
 This implies that the function  $f_{m,v}(s)$  decreases in
 $s\in(0,+\infty)$. By the Riemann-Lebesgue Theorem, it follows that $G_m(v)\to
 0$ as $||v||_2\to +\infty$. Hence $f_{m,v}(+\infty)=0$ for $v\ne 0$.
 The statement {\bf 2} is proved.

Let us prove the statement {\bf 3.}
 {\bf i}. From statement {\bf 2} it follows that for all $v\in\R^d$, $v\ne 0$,
 the function  $G_{d}(sv)$  decreases in $s\in[0,+\infty)$ and,
 hence, $0\le G_d(v)\le G_d(0)$. Therefore, $G_d(0)>0$
 (otherwise $G_d(v)\equiv 0$ on $\R^d$ $\Rightarrow$ $\varphi(\rho^2(y))\equiv 0$ on $\R^d$,
 that contradicts the condition  $\varphi(0)>0$).

 {\bf ii}. If, in addition, $G_d$ is real-analytic,
  then $\forall v\in\R^d$, $v\ne 0$, the function  $G_d(sv)$ strictly decreases on
$[0,+\infty)$. This can be proved by contraddiction. Let us assume
that, for some $v_0\in\R^d$ and $v_0\ne 0$, the function
 $G_d(sv_0)$ is constant on some interval
  $(\alpha,\beta)\subset (0,+\infty)$, $\alpha<\beta$. This would imply
  that $G_d$
  it is constant on $[0,+\infty)$ and
  $G_d(0)=\lim_{s\to+\infty}G_d(sv_0)=0$, which contradicts {\bf i}.
  Thus, $\forall v\in\R^d$, $v\ne 0$, the function  $G_d(sv)$ strictly decreases on
$[0,+\infty)$ and, hence, $G_d(v)>\lim_{s\to+\infty}G_d(sv)=0$. The
statement {\bf 3} is proved.

 Let us prove statement {\bf 4.} Let $v\in\s^{d-1}$ and
$f_{d,v}(s):=G_d(sv)$. From statements {\bf 2} and  {\bf 3}, it
follows that the function  $f_{d,v}(s)$  decreases on $[0,+\infty)$
and that $f_{d,v}(0)>0$. Obviously, $f_{d,v}(s)\in C^2(\R)$ and
$$
 f_{d,v}(s)=f_{d,v}(0)+\frac{ f''_{d,v}(0)}{2}\,s^2+o(s^2)\;,\;s\to 0\,,
$$
where $f''_{d,v}(0)=-\alpha_1(v)$. Note that  $f''_{d,v}(0)\le 0$,
otherwise the function  $f_{d,v}(s)$ strongly increases on $[0,c]$
for some $c>0$, which contradicts statement {\bf 2}. Thus,
$\alpha_1(v)\ge 0$ for all $v\in\s^{d-1}$.
 For $p>0$, the next  integral is constant on $\s^{d-1}$:
 $$
 \int_{\s^{d-1}} |(y,v)|^p \,d\sigma(v)\equiv
 c_{d,p}>0\;,\;y\in\s^{d-1}\;,
 $$ where  $d\sigma$, if $n\ge 2$, is the surface measure on $\s^{d-1}$
and $d\sigma(v)=\delta(v-1)+\delta(v+1)$, if $d=1$ (here $\delta
(v)$ - the Dirac measure with mass $1$ concentrated in the point
$v=0$).
 Therefore,
\begin{equation}\label{cdp}
\int_{\s^{d-1}} |(y,v)|^p \,d\sigma(v)=
 c_{d,p}||y||^p_2\;,\;y\in\R^{d}\,,\,p>0.
\end{equation}
Hence
$$
\int_{\s^{d-1}} \alpha_1(v)\,d\sigma(v)=c_{d,2}\,\beta_1\ge 0
$$
and $\alpha_1(v)\equiv 0$ on $\s^{d-1} \iff \beta_1=0$.\\
Let, in addition, $\beta_1>0$. Then
$f''_{d,v_0}(0)=-\alpha_1(v_0)<0$ for some $v_0\in\s^{d-1}$ and
\begin{equation}\label{psi_n}
\psi_n(t):=\left(\frac{G_d(\gamma_n\sqrt{h(t)}\,v_0)}{G_d(0)}\right)^n=(1+g_n(t))^n\in\Phi(E)\;,\;
\forall n\in\N\,,\,\gamma_n>0.
\end{equation}
Take
$$
  \gamma_n:=\left(-\frac{2f_{d,v_0}(0)}{f''_{d,v_0}(0)}\cdot\frac{\lambda}{n}\right)^{\frac 12}>0\;,\;\lambda>0.
$$
Obviously, $\gamma_n\to+0$ and
$$
 g_n(t)=\frac{f_{d,v_0}(\gamma_n\sqrt{h(t)})-f_{d,v_0}(0)}{f_{d,v_0}(0)}\sim
 \frac{f''_{d,v_0}(0)}{2f_{d,v_0}(0)}\cdot(\gamma_n\sqrt{h(t)})^2=-\,\frac{\lambda}{n}\cdot
 h(t)\;,\;n\to\infty\,.
$$
Therefore, $\psi_n(t)\to e^{-\lambda h(t)}$ and, hence,
 $e^{-\lambda h(t)}\in\Phi(E)$ for all $\lambda>0$.
 The statement {\bf 4} is proved.

 Let us prove the statement {\bf 5.}
 In this case  $G_d$ is real-analytic and
\begin{equation}\label{beta_k}
f^{(2k)}_{d,v}(0)=(-1)^k\alpha_k(v)\;,\;f^{(2k-1)}_{d,v}(0)=0\;,\;
 \int_{\s^{d-1}}
 \alpha_k(v)\,d\sigma(v)=c_{d,2k}\,\beta_k\;,\;k\in\N\,.
\end{equation}
  Therefore, $\forall
 v\in\s^{d-1}$ $\exists p\in\N$ so that
 $$
 f_{d,v}(s)=f_{d,v}(0)+\frac{ f^{(2p)}_{d,v}(0)}{(2p)!}\,s^{2p}+o(s^{2p})\;,\;s\to 0\,,
$$
where $f^{(2p)}_{d,v}(0)\ne 0$, otherwise the function
$f_{d,v}(0)\equiv f_{d,v}(s)\equiv f_{d,v}(+\infty)=0$ which
contradicts the inequality $G_d(0)>0$ (see statement {\bf 3}).
Hence, $f^{(2p)}_{d,v}(0)< 0$, otherwise the function $f_{d,v}(s)$
 strongly
increases on $[0,c]$ for some $c>0$, which contradicts statement
{\bf 2}. Thus the function $p(v)$, $v\in\s^{d-1}$, defines
correctly.

Let $v\in\s^{d-1}$ and $p=p(v)$. Take function \eqref{psi_n}, where
$v_0=v$
$$
  \gamma_n:=\left(-\frac{(2p)!f_{d,v_0}(0)}{f^{(2p)}_{d,v_0}(0)}\cdot\frac{\lambda}{n}\right)^{\frac{1}{2p}}>0\;,\;\lambda>0.
$$
Then $ g_n(t)\sim-\,\frac{\lambda}{n}\cdot h^p(t)$, $n\to\infty$.
Therefore
 $\psi_n(t)\to e^{-\lambda h^p(t)}$ and, hence,
 $e^{-\lambda h^p(t)}\in\Phi(E)$ for all $\lambda>0$.

 If $\alpha_k(v_0)\ne 0$ for some $v_0\in\s^{d-2}$, $k\in\N$,
  then $\alpha_k(v)\ne 0$ in some neighborhood of a point $v_0$
  and, hence, $p(v)\le p(v_0)$ in this neighborhood.
  Thus the function $p(v)$ is locally bounded on compact $\s^{d-1}$
  and, hence, $p(v)$ is  bounded on  $\s^{d-1}$.

 Let $m=\min\limits_{v\in\s^{d-1}}p(v)=p(v_0)$ for some
 $v_0\in\s^{d-1}$.
 Then $\alpha_m(v_0)\ne 0$ and for all $v\in\s^{d-1}$ equality
$$
 f_{d,v}(s)=f_{d,v}(0)+\frac{ f^{(2m)}_{d,v}(0)}{(2m)!}\,s^{2m}+o(s^{2m})\;,\;s\to 0
$$
  holds. Obviously  $(-1)^k\alpha_k(v)=f^{(2k)}_{d,v}(0)=0$, for all $1\le k<m$ (if $m\ge 2$),
  and $(-1)^m\alpha_m(v)=f^{(2m)}_{d,v}(0)\le 0$ (otherwise the function  $f_{d,v}(s)$ strongly increases on $[0,c]$
for some $c>0$ that contradicts a statement {\bf 2}). From
  \eqref{beta_k} follows that $\beta_k=0$ for all $1\le k<m$ (if $m\ge 2$)
   and $(-1)^m\beta_m<0$. Therefore $q=m$.

   The Theorem~\ref{th1} is proved.
\end{proof}

\begin{proof}[ {\bf Proof of Theorem \ref{th3}}]
 If $h(t)\equiv h(0)>0$ on $E$, then the implication 1) $\Rightarrow$ 2) is obvious.
 If $h(t)\not\equiv h(0)$ on $E$, then this implication follows from statement {\bf 4} of Theorem \ref{th1} for
 $\varphi(s)=e^{-s}\in C_{[0,+\infty)}\bigcap M_{(0,+\infty)}$.

 The reverse implication 2) $\Rightarrow$ 1) follows from Lemma \ref{le1}
 for $\varphi(s)=e^{-s}$,  equality
 $$
 \int_{\R^d}e^{-\frac{1}{2\sigma}||y||^2_2}\,e^{i(y,v)}\,dy=(2\pi \sigma)^{\frac d2}\,
 e^{-\frac{\sigma}{2}||v||^2_2}\;,\;v\in\R^d\,,\,\sigma>0\,,
 $$
 and  Bernstein-Widder's Theorem. The proof is completed.
 \end{proof}

Next Theorem \ref{th2} is an addition to Theorem \ref{th1} for the
case $\rho(x)=||x||_2$.
\begin{theorem}\label{th2}
Let the next conditions be satisfied:\\
1) $h\in C(E)$, $h(t)>0$  $\forall t\in E$ and $h(t)\not\equiv h(0)$ on $E$.\\
2) $\varphi\in C_{[0,+\infty)}$, $\varphi(0)>0$.\\
3)
 $
 K(x,t):=(h(t))^{-\frac d2}\varphi\left(\frac{||x||_2^2}{h(t)}\right)\in \Phi(\R^d\times E)
 $.\\
If
 $\int_{0}^{+\infty}|\varphi(u^2)|u^{m-1}\,du<+\infty$ for some
 natural $m\in[1,d]$ and $g_m$ is real-analytic,
 then the function $f_m(s):=s^{m-d}g_m(s)$ strictly decreases on
 $(0,+\infty)$ and $g_m(s)>0$ for all $s>0$.
  \end{theorem}
 \begin{proof}
 From Theorem \ref{th1} it follows that $f_m$ decreases on
 $(0,+\infty)$ and $f_m(s)\ge f_m(+\infty)=0$ for $s>0$.
 Since $f_m$ is real-analytic on $(0,+\infty)$,
  then   function  $f_m(s)$ strictly decreases on
$(0,+\infty)$.
 Otherwise  the function
 $f_m$ is constant on some interval
  $(\alpha,\beta)\subset (0,+\infty)$, $\alpha<\beta$, and, hence,
  it is constant on $(0,+\infty)$ and
  $f_m(s)=f_m(+\infty)=0$, $s>0$.
  Therefore, $G_m(v)=(2\pi)^{\frac{m}{2}}g_m(||v||_2)\equiv 0$ on
  $\R^m$.
 Hence, $\varphi(||x||_2^2)\equiv 0$ on  $\R^m$, which contradicts the condition $\varphi(0)>0$.
  Thus,  the function  $f_m$ strictly decreases on
$(0,+\infty)$ and, hence, $f_m(s)>f_m(+\infty)=0$ for all $s>0$.
 The Theorem \ref{th2} is proved.
 \end{proof}

\section{Some statements involving a versatile general covariance function}

Previous results can be generalized to the class of positive
definite functions built in \cite{Porcu:Gregori:Mateu2006} and used
for the purposes highlighted in Section 1.

\begin{lemma}
 Let $\varphi\in
C\left([0,+\infty)^n\right)$, $n\in\N$, and
$\int_{0}^{\infty}\dots\int_{0}^{\infty}
\left|\varphi(u_1^2,\ldots,u_n^2)\right|\prod_{k=1}^{n}u_k^{d_k-1}
\;d u_1\ldots d u_n <+\infty $ for some $d_k \in \mathbb{N}$,
$k=1,\ldots,n$.
 Let
$h_k,b_k \in C(E_k)$, $h_k$ strictly positive in their arguments for
all $k=1,\ldots,n$. Then
$$K(x_1,\ldots,x_n,t_1,\ldots,t_n) :=  \varphi \left ( \frac{\|x_1\|^2_2}{h_1(t_1)},\dots,\frac{\|x_n\|^2_2}{h_n(t_n)} \right )\, \prod_{k=1}^{n}b_k(t_k) \in \Phi(\R^{d_1} \times\ldots\times \R^{d_n} \times E_1 \times\ldots\times E_n )
$$ if, and only if,
   $$
    g_{d_1,\ldots,d_n} \left(s_1\sqrt{h_1(t_1)},\ldots,s_n\sqrt{h_n(t_n)}\right)\, \prod_{k=1}^{n}b_k(t_k)  \big ( h_k(t_k)\big )^{d_k/2}
   \in \Phi(E_1 \times\ldots\times E_n)
   $$
 for every $s_k \geq 0$, $k=1,\ldots,n$, where
$$g_{d_1,\ldots,d_n}(s_1,\ldots,s_n):= \int_{0}^{\infty}\ldots \int_{0}^{\infty} \varphi(u_1^2,\ldots,u_n^2) \prod_{k=1}^{n}u_k^{d_k-1} j_{d_k/2-1}(s_ku_k)\; du_1\ldots du_n .$$
\end{lemma}

\begin{proof}
The statement  can be proved in a similar way as Lemma \ref{le1}.
\end{proof}
Let $P_+^{\,n}$ be the set all finite nonnegative Borel  measures on
$[0,+\infty)^n$, $n\in\N$, and
$$\mathcal{L}^n:=\left\{\varphi(u_1,\dots,u_n)=\int_{0}^{\infty}\ldots \int_{0}^{\infty} e^{-(u_1v_1+\ldots+u_nv_n)}\; d\mu(v_1,\ldots, v_n)\;,\;\mu\in P_+^{\,n} \right\}\;.$$
 Obviously, $\mathcal{L}^1=C_{[0,+\infty)}\bigcap M_{(0,+\infty)}$ and
 $\prod_{k=1}^{n}\varphi_k(u_k)\in\mathcal{L}^n$ for every $\varphi_k\in\mathcal{L}^1$,  $k=1,\ldots,n$.
\begin{theorem}
Let $n\in\N$. For all $k=1,\ldots,n$, let $E_k$ be linear spaces,
$h_k$ strictly positive functions such that $h_k \in C(E_k)$ and
$d_k\in\N$. Then, the following statements are equivalent:
\begin{description}
  \item[1.]
  $K(x_1,\ldots,x_n,t_1,\ldots,t_n) :=  \varphi \left ( \frac{\|x_1\|^2_2}{h_1(t_1)},\dots,\frac{\|x_n\|^2_2}{h_n(t_n)} \right )\, \prod_{k=1}^{n} \big ( h_k(t_k)\big )^{-d_k/2}\in \Phi(\R^{d_1} \times\ldots\times \R^{d_n} \times E_1 \times\ldots\times E_n  )$
  $\forall\varphi\in\mathcal{L}^n$.
    \item[2.] $e^{-\lambda h_k(t_k)}\in \Phi(E_k)$, $\forall \lambda>0$, $k=1,\ldots,n$.
\end{description}
\end{theorem}
\begin{proof}
Let us prove the  implication $(1) \Longrightarrow (2) $. For every
fixed $k=1,\ldots,n$ in condition {\bf 1}, we take
$\varphi(u_1,\dots,u_n)=\varphi_k(u_k)$,
$\varphi_k\in\mathcal{L}^1$, and $t_i=0\in E_i$ for $i\ne k$. Then
   $\varphi_k\left(\frac{||x_k||_2^2}{h_k(t_k)}\right)\,(h_k(t_k))^{-d_k/2}\in \Phi(\R^{d_k}\times E_k)$
   $\forall \varphi_k\in\mathcal{L}^1$.
 By Theorem~\ref{th3} we get $e^{-\lambda h_k(t_k)}\in \Phi(E_k)$ $\forall \lambda>0$.

 Let us now prove the reverse implication. Let
   $e^{-\lambda h_k(t_k)}\in \Phi(E_k)$, $\forall \lambda>0$, $k=1,\ldots,n$.
    By Theorem~\ref{th3}, we have that
    $\varphi_k\left(\frac{||x_k||_2^2}{h_k(t_k)}\right)\,(h_k(t_k))^{-d_k/2}\in \Phi(\R^{d_k}\times E_k)$
   $\forall \varphi_k\in\mathcal{L}^1$, $k=1,\ldots,n$. We take
   $\varphi_k(u_k)=e^{-u_kv_k}$, $v_k\ge 0$.
   From definition of class $\mathcal{L}^n$ follows, that
   $\varphi \left ( \frac{\|x_1\|^2_2}{h_1(t_1)},\dots,\frac{\|x_n\|^2_2}{h_n(t_n)} \right )\, \prod_{k=1}^{n} \big ( h_k(t_k)\big )^{-d_k/2}\in \Phi(\R^{d_1} \times\ldots\times \R^{d_n} \times E_1 \times\ldots\times E_n  )$
  $\forall\varphi\in\mathcal{L}^n$.
\end{proof}


\end{document}